 \documentclass[smallabstract,smallcaptions]{dccpaper}

\usepackage{srcltx}
\usepackage{epsfig}
\usepackage{citesort}
\usepackage{amsmath}
\usepackage{amssymb}
\usepackage{color}
\usepackage{url}
\usepackage{pstricks}
\usepackage{pst-node}
\usepackage{pst-tree}
\usepackage{graphicx}
\usepackage[cp1252]{inputenc}
\usepackage{caption} 

\newtheorem{theorem}{Theorem}
\newenvironment{proof}{\noindent {\em Proof:}}{\hfill $\Box$ \\ }

\newlength{\figurewidth}
\newlength{\smallfigurewidth}

\begin{document}

\title
{\large
\textbf{Approximating Optimal Bidirectional Macro Schemes
\thanks{ The work reported in this article was supported by national funds
through
  Funda\c{c}\~ao para a Ci\^encia e Tecnologia (FCT) through projects NGPHYLO PTDC/CCI-BIO/29676/2017 and
    UID/CEC/50021/2019.
Funded in part by European Union's Horizon 2020 research and innovation
programme under the Marie Sklodowska-Curie grant agreement No 690941 (project
BIRDS). G.N.\ funded in part by Millennium Institute for Foundational Research
on Data (IMFD), Chile.}}
}

\author{%
Lu\'{i}s M. S. Russo$^{\ast}$,
Ana D. Correia$^{\ast}$,
Gonzalo Navarro$^{\dag}$,
Alexandre P. Francisco$^{\ast}$\\[0.5em]
{\small\begin{minipage}{\linewidth}\begin{center}
\begin{tabular}{ccc}
$^{\ast}$INESC-ID, Dept. of Computer & \hspace*{-0.2in} & $^{\dag}$Millennium Institute for \\
Science and
Engineering && Foundational Research on Data (IMFD), \\
Instituto Superior T\'ecnico && Dept. of Computer Science, \\
Universidade de Lisboa, Portugal. && University of Chile, Chile. \\
\url{luis.russo@tecnico.ulisboa.pt}
&& \url{gnavarro@dcc.uchile.cl} \\
\url{ana.duarte.correia@tecnico.ulisboa.pt} \\
\url{aplf@tecnico.ulisboa.pt} \\
\end{tabular}
\end{center}\end{minipage}}
}

\maketitle
\thispagestyle{empty}

\begin{abstract}
%
%
Lempel-Ziv is  an easy-to-compute member of a wide family of so-called
macro schemes; it restricts pointers to go in one direction only.
Optimal bidirectional
macro schemes are NP-complete to find, but they may provide much better
compression on highly repetitive sequences. We consider the problem of
approximating optimal bidirectional macro schemes. We describe a simulated
annealing algorithm that usually converges quickly. Moreover, in some cases,
we obtain bidirectional macro schemes that are provably a 2-approximation of
the optimal. We test our
algorithm on a number of artificial repetitive texts and verify that
it is efficient in practice and outperforms Lempel-Ziv, sometimes by a wide
margin.
\end{abstract}


\Section{Introduction} 

In 1976, Lempel and Ziv \cite{LZ76} proposed a technique to measure the
complexity of finite sequences that later became a popular compression
algorithm. It is a greedy left-to-right parse of the sequence into ``phrases''
that, at each step, extends the current phrase as much as possible as long as
the sequence contains another occurrence of the phrase starting before it.
Then, it adds one more symbol to the phrase (which makes it unique in
the sequence seen so far). Such a so-called Lempel-Ziv parse can be computed
in linear time \cite{RPE81}, which has made it a very popular compression
method.

Storer and Szymanski \cite{SS82} studied a much wider class of so-called
``macro schemes''. In particular, the smallest ``bidirectional macro scheme''
partitions the sequence into a sequence of phrases such that each phrase is
either an explicit symbol or it can be copied from somewhere else in the text,
as long as cycles are not introduced in the copying process. Such schemes can
produce parsings up to $\Theta(\log n)$ times smaller than
Lempel-Ziv~\cite{GNP18}
(on a text of length $n$), but unfortunately finding the optimal
bidirectional macro scheme is NP-complete \cite{Gal82}. This has hampered its
popularity.

In this paper we describe an algorithm to efficiently compute a small
(bidirectional) macro scheme. Our algorithm uses simulated annealing and
usually converges in a very small number of steps to a local minimum,
which, in some cases, is provably a 2-approximation to the smallest macro scheme for the sequence.
We show experimentally that the algorithm can obtain macro schemes that are
much smaller than those obtained by Lempel-Ziv, and that it is efficient in
practice thanks to the use of appropriate data structures.
To test the algorithm we devise a method for generating
highly repetitive strings whose optimal macro scheme is known. We also
experiment with families of highly compressible and uncompressible strings.

For practicality, we consider a slightly simpler class of macro schemes, which
makes them closer to the output of Lempel-Ziv parsings: the text is parsed into
a sequence of phrases, each of which has an explicit symbol at the end.

\def\s#1{[mnode=r]\psframebox*[fillstyle=solid,fillcolor=black,framearc=0.4]{\bf\textcolor{white}#1}}
\def\psrowhooki{\tiny \color{gray}}
\begin{figure}[t]
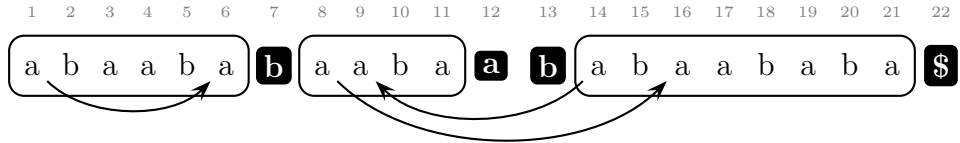

\begin{center}
\begin{psmatrix}[colsep=0.3,rowsep=0.2]
1& 2& 3& 4& 5& 6& 7& 8& 9& 10& 11& 12& 13& 14& 15& 16& 17& 18& 19& 20& 21& 22\\
a&b&a&a&b&a&\s{b}&a&a&b&a&\s{a}&\s{b}&a&b&a&a&b&a&b&a&\s{\$}\\
{
\psset{linearc=0.2}
\psset{nodesep=0.2}
\ncbox{2,1}{2,6}
\ncbox{2,8}{2,11}
\ncbox{2,14}{2,21}
}
{
\psset{nodesep=0.1}
\psset{arcangle=-45}
\psset{arrowsize=0.2}
\ncarc{->}{2,1}{2,6}
\ncarc{->}{2,8}{2,16}
\ncarc[arcangle=45]{->}{2,14}{2,9}
}
\end{psmatrix}
\end{center}
\vspace*{-5mm}
\caption{Example of a valid macro scheme.}
\label{fig:validMS}
\end{figure}

Figure~\ref{fig:validMS} shows an example of a valid macro scheme. An encoding
of this scheme is $(6,6,b),(16,4,a),(0,0,b),(9,8,\$)$, where
each tuple describes a phrase by giving a pointer to another position where
the phrase occurs,
the length of the phrase and the letter that ends the phrase. For technical
convenience, we
use a special terminator letter in the last phrase.
Notice that the first and second
phrases actually point forward in the string, which would be invalid in the Lempel-Ziv
encoding. This macro scheme is valid because all the letters can be decoded
without falling in loops.

\Section{The Problem} 
 Our challenge is to find a small valid macro scheme. Consider again the
macro scheme in Figure~\ref{fig:validMS}. This macro scheme is valid because
it is possible to decode every letter of the encoded string by using
the pointers in the encoding. For example to decode the letter at position 4 we can
follow the corresponding pointers to 9, 17 and finally to 12, which contains an explicit
letter.

Figure~\ref{fig:invalidMS} shows an invalid macro scheme, where
several letter positions are impossible to decode.
Consider, for example,
the letter at position 21. By following the pointers in the scheme we obtain the
sequence of positions $16, 11, 6, 11, \ldots$.
Hence the decoding process gets stuck in a loop that maps between 11 and 6. We will thus require
a data structure that can detect this problem, and a procedure to resolve it.

\begin{figure}[b]
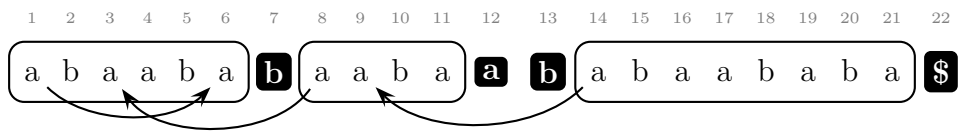

\begin{center}
\begin{psmatrix}[colsep=0.3,rowsep=0.2]
1& 2& 3& 4& 5& 6& 7& 8& 9& 10& 11& 12& 13& 14& 15& 16& 17& 18& 19& 20& 21& 22\\
a&b&a&a&b&a&\s{b}&a&a&b&a&\s{a}&\s{b}&a&b&a&a&b&a&b&a&\s{\$}\\
{
\psset{linearc=0.2}
\psset{nodesep=0.2}
\ncbox{2,1}{2,6}
\ncbox{2,8}{2,11}
\ncbox{2,14}{2,21}
}
{
\psset{nodesep=0.1}
\psset{arcangle=-45}
\psset{arrowsize=0.2}
\ncarc{->}{2,1}{2,6}
\ncarc[arcangle=55]{->}{2,8}{2,3}
\ncarc[arcangle=45]{->}{2,14}{2,9}
}
\end{psmatrix}
\end{center}
\vspace*{-5mm}
\captionsetup{belowskip=0pt}
\caption{Example of an invalid macro scheme.}
\label{fig:invalidMS}
\end{figure}

An important issue in determining the smallest macro scheme is how to select
the phrases. Testing every possible configuration is unfeasible. Our
approach will instead seek to alter a given configuration by searching
for good nearby configurations. This will amount to merging and splitting
phrases.
Figure~\ref{fig:smallerMS} shows a valid configuration
that can be obtained by merging the first two phrases in the
configuration of Figure~\ref{fig:validMS}.
Merging phrases reduces their total amount by one, whereas splitting
does the opposite. Still, splitting phrases may be essential to avoid loops or
as a stepping stone to better configurations.

\begin{figure}[t]
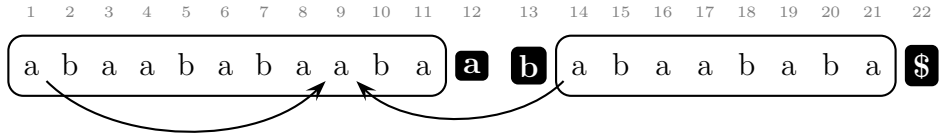

\begin{center}
\begin{psmatrix}[colsep=0.3,rowsep=0.2]
1& 2& 3& 4& 5& 6& 7& 8& 9& 10& 11& 12& 13& 14& 15& 16& 17& 18& 19& 20& 21& 22\\
a&b&a&a&b&a&b&a&a&b&a&\s{a}&\s{b}&a&b&a&a&b&a&b&a&\s{\$}\\
{
\psset{linearc=0.2}
\psset{nodesep=0.2}
\ncbox{2,1}{2,11}
\ncbox{2,14}{2,21}
}
{
\psset{nodesep=0.1}
\psset{arcangle=-45}
\psset{arrowsize=0.2}
\ncarc{->}{2,1}{2,9}
\ncarc[arcangle=45]{->}{2,14}{2,9}
}
\end{psmatrix}
\end{center}
\vspace*{-5mm}
\caption{Example of a smaller valid macro scheme.}
\label{fig:smallerMS}
\end{figure}
Like Lempel-Ziv, our algorithm also requires a data structure to locate
identical copies of the string inside the phrase; in particular we use a suffix
array \cite{MM93}.

\Section{An Annealing Algorithm} 
In this section we describe our general approach. We use the simulated
annealing technique \cite{LA87}, where each configuration is a state and a neighbor
state can be obtained by merging or splitting phrases. A transition
that successfully merges two phrases is always accepted. A transition that
splits phrases
may be accepted or rejected, depending on the current temperature $(t)$,
the increase in the number of phrases ($\delta$) and a random number $(p)$ chosen uniformly from
$[0,1]$. If Equation~(\ref{eq:1}) holds the transition is accepted, otherwise its rejected.
\begin{equation}
  \label{eq:1}
  \delta \leq - t \ln p
\end{equation}

At each step the algorithm chooses a phrase uniformly at random
and tries to merge it with the next phrase. For example we can
choose to merge the first and second phrases in the configuration in Figure~\ref{fig:validMS}.
In a successful attempt we can obtain the configuration in Figure~\ref{fig:smallerMS}. To
determine this configuration we locate another copy of the substring
$abaababaaba$. With the suffix array, we efficiently find that this
string occurs at positions $1$ and $9$.
Choosing a pointer from position $1$ to position $1$, would trivially lead to loops in the decoding
process. Hence
these kinds of pointers are always rejected. Fortunately, pointing to position
$9$ yields
a valid macro scheme.

\begin{figure}[b]
\begin{center}
\begin{pspicture}[showgrid=false](13,2)
\rput[bl](0,0){
\begin{psmatrix}[colsep=0.3,rowsep=0.2]
1& 2& 3& 4& 5& 6& 7& 8& 9& 10& 11& 12& 13& 14& 15& 16& 17& 18& 19& 20& 21& 22\\
a&b&\s{a}&a&b&a&\s{b}&a&a&b&a&\s{a}&\s{b}&a&b&a&a&b&a&b&a&\s{\$}\\
{
\psset{linearc=0.2}
\psset{nodesep=0.2}
\ncbox{2,1}{2,2}
\ncbox{2,4}{2,6}
\ncbox{2,8}{2,11}
\ncbox{2,14}{2,21}
}
{
\psset{nodesep=0.1}
\psset{arcangle=45}
\psset{arrowsize=0.2}
\ncarc{->}{2,1}{2,12}
\ncarc{->}{2,4}{2,12}
\ncarc[arcangle=55]{->}{2,8}{2,3}
\ncarc[arcangle=45]{->}{2,14}{2,9}
}
\end{psmatrix}
}
\end{pspicture}
\end{center}
\vspace*{-5mm}
\captionsetup{belowskip=0pt}
\caption{Example of a complex phrase merge.}
\label{fig:semivalidMS}
\end{figure}
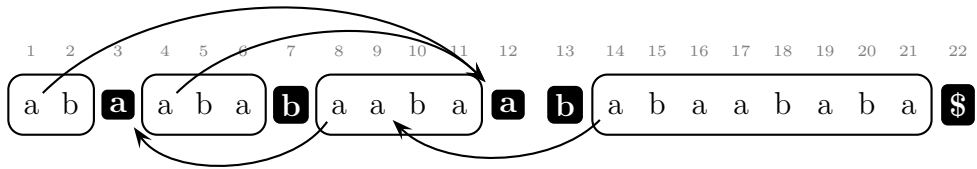


As a more involved example, assume that the current configuration is the one
in Figure~\ref{fig:semivalidMS} and that we decide to merge the first and second
phrases. We now need to select a pointer for the new phrase.
We use the suffix array to search for the string $abaaba$.
The resulting positions are $1,6,9,14$, where $1$ corresponds to the trivial loop
and is therefore excluded. We then select between $6$, $9$ and $14$ uniformly at random.
If we end up selecting $6$, then the configuration is the one presented in Figure~\ref{fig:invalidMS}.
As explained before, this configuration is invalid, thus some additional process is necessary
to obtain a valid configuration. First we try to sample again from the possible
pointers, $6, 9$ and $14$. If a valid configuration is obtained, then the
transition
is passed to the simulated annealing process and subsequently accepted. If, after
4 attempts, the process keeps on generating invalid macro schemes, we proceed to
splitting phrases.


Let us assume that the configuration in Figure~\ref{fig:invalidMS} was obtained
after 4 failed attempts. As illustrated, the letter at position $21$ cannot be decoded, because
it gets captured in a loop involving $11$ and $6$. However, when we select a new
pointer, we only need to try to decode the letters inside the newly created
phrase, and not
all the letters in the text. This means that $21$ is not tested by this process.
Instead, position $1$ is, and it will expose the underlying issue, because it
also gets captured into the cycle formed by $11$ and $6$.
We break this cycle by choosing uniformly between the positions in the cycle,
in our case $11$ and $6$. The selected position becomes an explicit letter.
Note that we do not consider position $1$: even
though it was the position that revealed the loop, it is not inside the
cycle and therefore cannot be selected. Note that solving
the cycle will solve the problem for position $1$ and also for other positions.
For example the letter at position
$21$ will also become decodable. The selected position becomes an
explicit letter and thus
splits the phrase that contains it. For example, if $11$ is selected, the resulting
configuration is shown in Figure~\ref{fig:splitMS}.

\begin{figure}[b]
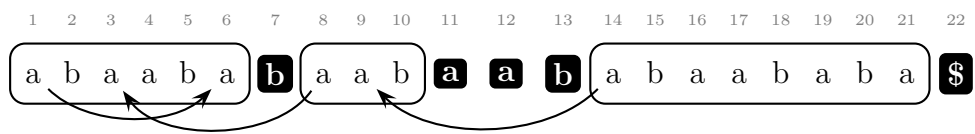

\begin{center}
\begin{psmatrix}[colsep=0.3,rowsep=0.2]
1& 2& 3& 4& 5& 6& 7& 8& 9& 10& 11& 12& 13& 14& 15& 16& 17& 18& 19& 20& 21& 22\\
a&b&a&a&b&a&\s{b}&a&a&b&\s{a}&\s{a}&\s{b}&a&b&a&a&b&a&b&a&\s{\$}\\
{
\psset{linearc=0.2}
\psset{nodesep=0.2}
\ncbox{2,1}{2,6}
\ncbox{2,8}{2,10}
\ncbox{2,14}{2,21}
}
{
\psset{nodesep=0.1}
\psset{arcangle=-45}
\psset{arrowsize=0.2}
\ncarc{->}{2,1}{2,6}
\ncarc[arcangle=55]{->}{2,8}{2,3}
\ncarc[arcangle=45]{->}{2,14}{2,9}
}
\end{psmatrix}
\end{center}
\vspace*{-5mm}
\captionsetup{belowskip=0pt}
\caption{Example of a phrase split.}
\label{fig:splitMS}
\end{figure}

Splitting a phrase is simpler than merging because we do not
need to select new pointers. In our example, the phrase that got split was
the second. This division did not produce two sub-phrases, only the left one.
Left sub-phrases always retain their pointer, in this case to position $3$.
If there was a right sub-phrase
it would point to position $7=3+4$, where $3$ is the original
pointer and $4$ the size of the left sub-phrase.

 Note that the configuration
in Figure~\ref{fig:splitMS} is still not a valid macro scheme: several positions still form cycles, for example
$4$ and $9$. The remaining cycles will be identified by trying to decode the letters
in the merged phrase, that is, positions from 1 to 6. In particular the cycle $4,9$ is
found by checking position $4$. In total there are four cycles in the configuration
of Figure~\ref{fig:invalidMS}, which means adding four phrases to that configuration.
Since this process started by merging two phrases, the overall difference in
the number of phrases is $3$. Hence, this transition is passed to the simulated
annealing algorithm, which decides whether to accept or reject the
transition, according to Equation~(\ref{eq:1}).

\Section{Data Structures and Optimizations} 
\label{sec:details}

In this section we discuss some details concerning the data structures
that we used for the implementation.

\SubSection{The Link Cut Tree Data Structure}

Checking whether any of the symbols of a merged phrase can fall into a loop
may take considerable time, since the decompressing paths can be long.
Instead, we use the link cut tree data structure \cite{ST83}, which
detects loops in only $O(\log n)$ amortized time per letter.
Since, in a valid macro scheme, it is possible to decode the letters at every
position, the decoding paths can be represented as a forest, where the roots
of the trees correspond to the positions that store explicit letters.

\begin{figure}[tb]
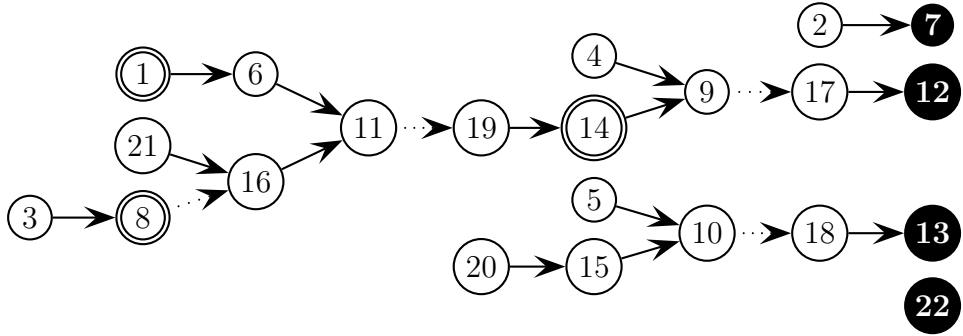

\begin{center}
{
\psset{arrowsize=0.3}
\psset{arrows=<-}
\psset{fillcolor=black}
\def\dedge{\ncline[linestyle=dotted]}

\psset{treesep=0.2}

\pstree[treemode=L,levelsep=0]{\Tp}{

\pstree[levelsep=1.5]{\Tcircle*[edge=none]{\textcolor{white}{\bf 7}}}{\Tcircle{2}}

\pstree[levelsep=1.5]{\Tcircle*[edge=none]{\textcolor{white}{\bf 12}}}{\pstree{\Tcircle{17}}{
\pstree{\Tcircle[edge=\dedge]{9}}{
\Tcircle{4}
\pstree{\Tcircle[doubleline=true]{14}}
{\pstree{\Tcircle{19}}{\pstree{\Tcircle[edge=\dedge]{11}}{
\pstree{\Tcircle{6}}{\Tcircle[doubleline=true]{1}}
\pstree{\Tcircle{16}}
{\Tcircle{21} \pstree{\Tcircle[edge=\dedge,doubleline=true]{8}}{\Tcircle{3}}}}}}}}}

\pstree[levelsep=1.5]{\Tcircle*[edge=none]{\textcolor{white}{\bf 13}}}{\pstree{\Tcircle{18}}{
\pstree{\Tcircle[edge=\dedge]{10}}
{\Tcircle{5}
\pstree{\Tcircle{15}}{\Tcircle{20}}}}}

\Tcircle*[edge=none]{\textcolor{white}{\bf 22}}
}
}
\caption{A forest representing the decoding paths of text positions of the
  macro scheme in Figure~\ref{fig:validMS}. Black circles represent
  positions that contain explicit letters; they are also roots. Double
  circles are used for positions that mark the beginning of blocks. Edges
  leaving double circles represent explicit links. Dotted edges are used to
  highlight edges that need to be altered when attempting to represent the
  scheme of Figure~\ref{fig:invalidMS}.}
\label{fig:decodeP}
\end{center}
\end{figure}

Figure~\ref{fig:decodeP} shows the forest corresponding to the macro scheme
of Figure~\ref{fig:validMS}.
Notice that this representation contains all the pointers in the macro
scheme. We have that $1$ points to $6$, that $8$ points to $16$, and that $14$
points to $9$. These starting positions of the phrases are shown with double
line circles in the figure. Hence the link cut tree representation contains all
the
information in the macro scheme except for the explicit letters. Moreover, the
forest contains all the implicit links that result from phrase pointers, for
example the pointer of the second phrase, from position $8$ to $16$, also induces
the implicit links $(9,17);(10,18);(11,19)$, shown as dotted arrows.

The link cut tree data structure supports edge insertion and removal, provided
the representation remains a forest at all times.  Let us discuss how this
structure changes when phrases are split or merged. Splitting phrases is
simple and efficient. Consider the configuration in Figure~\ref{fig:yas},
which results from splitting the configuration in Figure~\ref{fig:validMS} by
adding the letter at position $14$. In the link cut tree representation this
amounts to removing the edge that links $14$ in the path to its parent,
that is, the edge $(14,9)$. In general, splitting a phrase requires cutting a
single edge.

\begin{figure}[b]
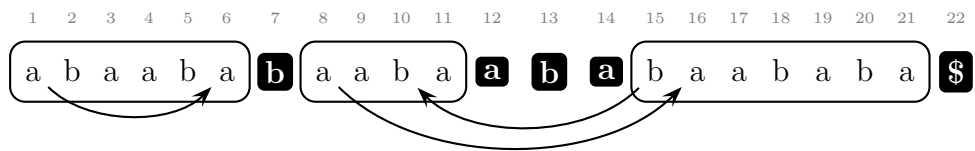

\begin{center}
\begin{psmatrix}[colsep=0.3,rowsep=0.2]
1& 2& 3& 4& 5& 6& 7& 8& 9& 10& 11& 12& 13& 14& 15& 16& 17& 18& 19& 20& 21& 22\\
a&b&a&a&b&a&\s{b}&a&a&b&a&\s{a}&\s{b}&\s{a}&b&a&a&b&a&b&a&\s{\$}\\
{
\psset{linearc=0.2}
\psset{nodesep=0.2}
\ncbox{2,1}{2,6}
\ncbox{2,8}{2,11}
\ncbox{2,15}{2,21}
}
{
\psset{nodesep=0.1}
\psset{arcangle=-45}
\psset{arrowsize=0.2}
\ncarc{->}{2,1}{2,6}
\ncarc{->}{2,8}{2,16}
\ncarc[arcangle=45]{->}{2,15}{2,10}
}
\end{psmatrix}
\end{center}
\vspace*{-5mm}
\captionsetup{belowskip=0pt}
\caption{Another example of a phrase split.}
\label{fig:yas} 
\end{figure}

Merging phrases requires more extensive modifications to the tree structure. In
particular, changing the pointer of a phrase implies altering all the
induced edges we mentioned above. Consider for example that we want to change
the configuration in Figure~\ref{fig:validMS} to that of
Figure~\ref{fig:invalidMS}. This requires changing all the pointers of the second
phrase. We first consider the positions inside this phrase, that is,
$8,9,10$, and $11$.
We cut the edges leaving these nodes, so they become roots in their trees.
These edges are drawn with dotted lines in Figure~\ref{fig:decodeP}. Then we
need to add the new edges (8,3);(9,4);(10,5);(11,6). However, it is necessary to
check if
this change does not introduce a cycle into the forest. This is supported by the
link cut tree data structure in $O(\log n)$ amortized time. So we first check if
there is a path from $3$ to $8$. In fact, there is a direct edge, so it is not
possible to add the edge $(8,3)$, because
it would result in an invalid macro scheme. The link cut tree data structure
supports selecting an edge from this path in $O(\log n)$ amortized time, which
combined with the cut operation can be used to implement the procedure
that splits a phrase that
contains a position in the underlying cycle. A similar process is used for the
remaining edges.



\SubSection{Suffix Arrays}

We use suffix arrays to determine the lexicographic range of all the
occurrences of a given phrase. In general this operation requires $O(m\log n)$
time for a phrase of size $m$. We use two optimizations. First, we cache the
searches by storing the resulting suffix array intervals.
When a phrase is split, this information is discarded. When two
phrases get merged we combine the two intervals in $O(\log n)$ time by using
the inverse suffix array and a binary search.

\SubSection{Optimizing the Simulated Annealing}


Another important optimization of our algorithm is related to the phrases that
cannot be merged with the next phrase because they result in a unique
substring. As discussed, this kind of
transitions is always rejected, because they induce a trivial loop. Once this
is detected for a given phrase, there is no point in reconsidering the phrase in
a future iteration, so the phrase gets removed from a list of admissible
phrases. The phrase selection procedure selects from this list, instead of
from all the
existing phrases. This speeds up the algorithm by skipping redundant steps. It
does require some maintenance, however. When a phrase is split, its two sub-phrases need to
be inserted in the list. Moreover, the phrase before the one being split also
needs to be re-inserted into the list, in case it is not present already. This
list should not contain repetitions, therefore we store it with a binary search
tree. An important side effect of this approach occurs when this list becomes
empty. In this case, the algorithm is stuck in a local minimum, which might go
unnoticed otherwise. Notice that the algorithm would still terminate as the
annealing temperature decreases, however no improvement would result from the
extra computation. This particular kind of minimum has important properties,
as we show next.

\Section{Approximation Ratio}

\begin{figure}[b]
\begin{center}
\begin{pspicture}[showgrid=false](14,1.8)
\rput[bl](0,0){
\begin{psmatrix}[colsep=0.1,rowsep=0.16]
\\
a&a&\s{a}&a&a&\s{b}&a&a&\s{a}&b&b&\s{a}&a&b&\s{a}&b&a&\s{a}&b&b&\s{b}&a&b&\s{a}&b&b&\s{a}&b&b&\s{b}&b&
\s{\$}\\
a&a&a&a&a&\s{b}&a&a&a&b&b&\s{a}&a&b&a&b&a&\s{a}&b&b&b&a&b&\s{a}&b&b&a&b&b&\s{b}&b& \s{\$}\\
a&a&a&a&\s{a}&b&a&a&\s{a}&b&b&a&\s{a}&b&a&b&a&\s{a}&b&b&b&\s{a}&b&a&b&b&\s{a}&b&b&\s{b}&b&\s{\$}\\
\end{psmatrix}
}
\ncline[linewidth=.5,nodesep=-0.2,linecolor=black!10]{c-c}{2,6}{4,6}
\ncline[linewidth=.5,nodesep=-0.2,linecolor=black!10]{c-c}{2,12}{4,12}
\ncline[linewidth=.5,nodesep=-0.2,linecolor=black!10]{c-c}{2,18}{4,18}
\ncline[linewidth=.5,nodesep=-0.2,linecolor=black!10]{c-c}{2,24}{4,24}
\ncline[linewidth=.5,nodesep=-0.2,linecolor=black!10]{c-c}{2,30}{4,30}
\rput[bl](0,0){
\begin{psmatrix}[colsep=0.1,rowsep=0.16]
\\
a&a&\s{a}&a&a&\s{b}&a&a&\s{a}&b&b&\s{a}&a&b&\s{a}&b&a&\s{a}&b&b&\s{b}&a&b&\s{a}&b&b&\s{a}&b&b&\s{b}&b&
\s{\$}\\
a&a&a&a&a&\s{b}&a&a&a&b&b&\s{a}&a&b&a&b&a&\s{a}&b&b&b&a&b&\s{a}&b&b&a&b&b&\s{b}&b& \s{\$}\\
a&a&a&a&\s{a}&b&a&a&\s{a}&b&b&a&\s{a}&b&a&b&a&\s{a}&b&b&b&\s{a}&b&a&b&b&\s{a}&b&b&\s{b}&b&\s{\$}\\
\end{psmatrix}
}
\end{pspicture}
\end{center}
\vspace*{-5mm}
\captionsetup{belowskip=0pt}
\caption{The top string shows a local minima macro scheme such that merging
adjacent phrases creates unique substrings. The middle configuration is obtained
by merging pairs of consecutive phrases, it is not a macro scheme and
therefore not an optimal configuration. The bottom configuration shows an
optimal macro scheme.}
\label{fig:Bruijn}
\end{figure}

Given that our algorithm never starts by splitting phrases, it may get
stuck in a local minimum. In fact, this occurred in  our experiments.
However, the particular structure of the minima turns out to be relevant.
In these minima merging any two consecutive phrases results in a phrase
that is unique, that is, it occurs only at its position. We will now prove that
such a configuration is a 2-approximation to the optimal macro scheme.

\begin{theorem}
Any configuration where every pair of consecutive phrases is unique is a
2-approximation to the optimal macro scheme of the sequence.
\end{theorem}
\begin{proof}
Consider the concatenation of any two consecutive phrases, which by hypothesis
is unique in the text. Such text substring cannot be inside a phrase of any
macro scheme, because in that case it should occur elsewhere. Thus, every two
consecutive phrases of our configuration must contain a boundary in any macro
scheme.
\end{proof}

Figure~\ref{fig:Bruijn} shows such a configuration and illustrates the
approximation argument. Its top string shows a valid macro scheme of 11 phrases
for the string, which moreover is a local minimum with the property that
merging any two consecutive phrases results in a unique substring. The
actual pointers are not relevant for this example. In the middle we consider
the configuration where phrase 1 is merged with phrase 2, phrase 3 is merged
with phrase 4, and so on. Since every substring in the middle is unique, there
must be a phrase in any macro scheme that ends within that string. The bottom
configuration illustrates this condition with a macro scheme of 8 phrases.

We can prove an even stronger result, related to {\em string attractors}
\cite{KP18}. An attractor is a set $\Gamma$ of text positions such that
any text substring must have a copy containing a position in $\Gamma$. It
is shown that the size $\gamma$ of the smallest attractor is a lower bound
to the size of any macro scheme. Further, finding $\gamma$ is NP-complete.
While it is not known whether we can always encode a text in $O(\gamma)$
space, we show that our approximation also applies to the smallest attractor.

\begin{theorem}
For any configuration where every pair of consecutive phrases is unique, the
set of the final phrase positions (i.e., the positions of the explicit symbols)
is an attractor of size at most $2\gamma$.
\end{theorem}
\begin{proof}
First, the set is an attractor because, by definition of macro scheme, any
substring that is completely inside a phrase must have another occurrence
containing an explicit symbol position. To see that its size is at most
$2\gamma$, consider again
the concatenation of any two consecutive phrases, which by hypothesis
is unique in the text. Therefore any attractor must contain a position inside
the phrase. If there is an odd number of phrases, then there must also be an
attractor position at the end of the text to cover the terminator \$.
\end{proof}

Even though this condition is not always attained, it does occurs several
times for some classes of strings. For those cases, it is a considerable
improvement to the $O(\log n)$ approximation provided by the
Lempel-Ziv~\cite{GNP18} algorithm.

 \Section{Experimental Results}
We implemented our algorithm to test its performance. We tested the convergence
speed with Fibonacci, Thue-Morse and binary de Bruijn sequences, as well as on
strings obtained from a generator we developed for this purpose. Fibonacci
sequences are binary strings defined as $F_1=b$, $F_2=a$, and $F_{n+2} =
F_{n+1}\cdot F_n$. They have macro schemes of size $3$ (using our
symbol-terminated kind of phrases) \cite{GNP18}, which is optimal with a
binary alphabet.
Thue-Morse sequences
are strings defined as $T_0 = 0$ and then $T_{n+1} = T_n \cdot \overline{T_n}$,
where $\overline{T_n}$ means complementing all the bits of $T_n$.
Their optimal macro scheme size is unknown, but a lower bound is the number of
distinct substrings of size $\ell$ divided by $\ell$, for any $\ell$
\cite{KNP19}. This is between\footnote{See
{\tt https://fr.wikipedia.org/wiki/Suite\_de\_Prouhet-Thue-Morse} and
{\tt https://oeis.org/A005942}.} $3$ and $10/3$, so we take $3$ as a lower bound.
Finally, the binary de Bruijn sequence of order $t$ contains all the distinct
substrings of length $t$ and is of minimum length, $2^t + t - 1$. Therefore~\cite{KNP19},
a lower bound to the size of any macro scheme is $1+(2^t/t)$.

Our generator
chooses an alphabet size $d$ and builds a text with a bidirectional
macro scheme of size $d$, which must be optimal because $d$ is a lower bound.
We put the distinct characters at random, as phrase
terminators, then define the sources of the phrases at random, and check that
the scheme is valid. Any resulting valid scheme is then a text whose smallest
macro scheme is of known size.

Figure~\ref{fig:iPlots} shows the number of iterations of the simulated
annealing algorithm versus the number of phrases $k$ in the obtained
macro scheme, aggregated over 100 steps. We also show the
number of Lempel-Ziv phrases and the size of
the smallest macro scheme, or a lower bound if it is not known.

\begin{figure}[tb]
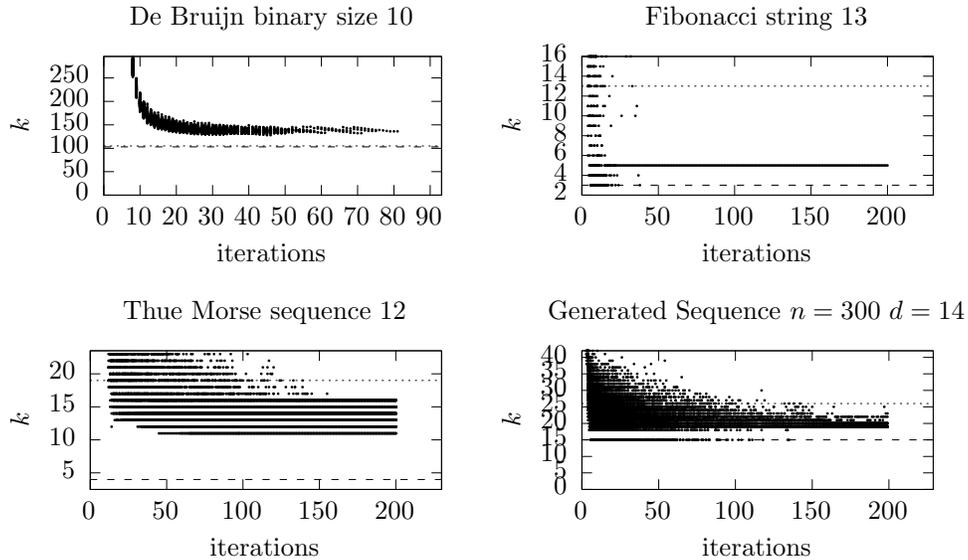

\centering
 \begin{minipage}[b]{.42\textwidth}
\footnotesize
 \centering \input{bruijn10.tex}
  \end{minipage}
 \begin{minipage}[b]{.42\textwidth}
\footnotesize
 \centering \input{fibonacci13.tex}
  \end{minipage}
 \begin{minipage}[b]{.42\textwidth}
\footnotesize
 \centering \input{thueMorse12.tex}
  \end{minipage}
 \begin{minipage}[b]{.42\textwidth}
\footnotesize
 \centering \input{generatorN300d14.tex}
  \end{minipage}
  %
\caption{Iterations of our algorithm ($\times 100$) versus number $k$ of
obtained phrases. The dotted line is the size of the Lempel-Ziv parse and
the dashed one the size of the optimal macro scheme (or a lower bound if
unknown).}
\label{fig:iPlots}
\end{figure}

Except on de Bruijn sequences, our algorithm obtains configurations that
require much fewer phrases than the Lempel-Ziv parse. In fact, in several executions
our algorithm obtains the optimal size. Except for the Thue-Morse strings
the algorithm was able to achieve the 2-approximation condition on some
runs.
For the
Thue-Morse strings the minimum value obtained by the algorithm was 11, but
this did not achieved the 2-approximation condition. The known lower bound
for this sequence was 4, but this may be below optimal.


The de Bruijin sequences did obtain the 2-approximation condition. In fact
the ratio is even better because the minimum size
is 103 and the points obtained are below 200. The ratio is closer
to 4/3, which is expected on average for this kind of sequences. To deduce
this factor notice that the string in Figure~\ref{fig:Bruijn} is a prefix
of a de Bruijn sequence. In this string any substring of size $5$ is
unique. Almost all the binary strings of size $5$ occur as substrings.
The 2-approximation is the worst
case. Trying to merge two blocks of size 2 yields a substring of size 5
that is unique. Merging two blocks of size $1$ or one block of size $1$ and
one block of size $2$ is always possible. This means that the resulting
blocks are essentially random, with sizes ranging from $2$ to $4$. In
general this amounts to choosing random numbers uniformly from
$[1/2,1]$. The resulting expected value is $3/4$, thus explaining the $4/3$
approximation.

For Fibonacci strings and generated sequences our algorithm quickly reaches
the optimum number of phrases. For the Thue-Morse sequences, our algorithm
produces macro schemes much smaller than Lempel-Ziv, albeit not a
guaranteed 2-approximation. For the generated strings we did achieve the
2-approximation condition some of the time, notice that when this condition
is obtained we terminate the algorithm. Otherwise the algorithm terminates
by reaching its maximum number of iterations. This means that the
executions that reach a 2-approximation stop yielding data points after
reaching the condition, thus thinning the cloud of points. The data points
that fade away are the ones that achieve the 2-approximation condition. For
the De Bruijn sequences this happens on all executions. For the generated
sequences several executions obtain the 2-approximation condition. Also
several iterations obtain less phrases than the Lempel-Ziv, notice the
points below the dotted line. Notice the trend for most of the points to
group below this line.


\Section{Conclusions and Further Work} 

We have shown that the smallest macro scheme \cite{SS82}, an
NP-hard-to-compute measure of compressibility, can be practically
approximated, for some highly repetitive families of strings. On most of
our tests we obtain much better approximations than the popular Lempel-Ziv
algorithm~\cite{LZ76}. This opens the door to stronger compression schemes
of highly repetitive sequences.

Our future steps are to devise a more practical version of our compression
algorithm. While
we have managed to make it practical, running over large files is still a
challenge. Because the link cut tree data structure uses one node
for each letter in the text. Although the memory usage is linear in the
size of the file it is a factor of more than $32$. We plan to reduce this
overhead by storing fewer nodes, but still supporting the necessary
operations. Moreover we also store the suffix array and the inverse suffix
array, in plain form, which also uses $16$ times the space of the
file. Using a compressed representation will require less space; a plethora
of such representations is now available~\cite{NMacmcs06}.
We aim to reduce the size of both these structures so that the amount
of extra space necessary to obtain the smallest macro scheme becomes
sub-linear in the size of the file to compress. Also the initialization of
the list of admissible phrases will have to be delayed until it is
small enough, otherwise it may require as much space as the previous data
structures. A simple solution is to initialize our algorithm with a
Lempel-Ziv parsing.


\Section{References}
\bibliographystyle{IEEEbib}
\bibliography{main}

\begin{thebibliography}{10}

\bibitem{LZ76}
A.~Lempel and J.~Ziv,
\newblock ``On the complexity of finite sequences,''
\newblock {\em IEEE Transactions on Information Theory}, vol. 22, no. 1, pp.
  75--81, 1976.

\bibitem{RPE81}
M.~Rodeh, V.~R. Pratt, and S.~Even,
\newblock ``Linear algorithm for data compression via string matching,''
\newblock {\em Journal of the {ACM}}, vol. 28, no. 1, pp. 16--24, 1981.

\bibitem{SS82}
J.~A. Storer and T.~G. Szymanski,
\newblock ``Data compression via textual substitution,''
\newblock {\em Journal of the {ACM}}, vol. 29, no. 4, pp. 928--951, 1982.

\bibitem{GNP18}
T.~Gagie, G.~Navarro, and N.~Prezza,
\newblock ``On the approximation ratio of {L}empel-{Z}iv parsing,''
\newblock in {\em {Proc. 13th Latin American Symposium on Theoretical
  Informatics (LATIN)}}, 2018, pp. 490--503.

\bibitem{Gal82}
J.~K. Gallant,
\newblock {\em String Compression Algorithms},
\newblock Ph.D. thesis, Princeton University, 1982.

\bibitem{MM93}
U.~Manber and G.~Myers,
\newblock ``Suffix arrays: a new method for on-line string searches,''
\newblock {\em SIAM Journal on Computing}, vol. 22, no. 5, pp. 935--948, 1993.

\bibitem{LA87}
P.J.M. van Laarhoven and E.H.L. Aarts,
\newblock {\em Simulated Annealing: Theory and Applications}, vol.~37 of {\em
  Mathematics and Its Applications},
\newblock Springer, 1987.

\bibitem{ST83}
D.~D. Sleator and R.~E. Tarjan,
\newblock ``A data structure for dynamic trees,''
\newblock {\em Journal of Computer and Systems Sciences}, vol. 26, no. 3, pp.
  362--391, 1983.

\bibitem{KP18}
D.~Kempa and N.~Prezza,
\newblock ``At the roots of dictionary compression: string attractors,''
\newblock in {\em Proc. 50th Annual Symposium on Theory of Computing (STOC)},
  2018, pp. 827--840.

\bibitem{KNP19}
T.~Kociumaka, G.~Navarro, and N.~Prezza,
\newblock ``Towards a definitive measure of repetitiveness,''
\newblock {\em CoRR}, vol. 1910.02151, 2019.

\bibitem{NMacmcs06}
G.~Navarro and V.~M{\"a}kinen,
\newblock ``Compressed full-text indexes,''
\newblock {\em ACM Computing Surveys}, vol. 39, no. 1, pp. article 2, 2007.

\end{thebibliography}

\end{document}